\documentclass{amsart}

\usepackage{amssymb}

\usepackage[all]{xy}

\theoremstyle{plain}
\newtheorem{theorem}{Theorem}
\newtheorem{proposition}[theorem]{Proposition}
\newtheorem{lemma}[theorem]{Lemma}
\newtheorem{corollary}[theorem]{Corollary}

\theoremstyle{definition}
\newtheorem{example}{Example}
\newtheorem{definition}{Definition}
\newtheorem{remark}{Remark}

\begin{document}

\title{Normal art galleries: wall in - all in}

\author{Zoran {\v S}uni{\'c}}
\thanks{Partially supported by the National Science Foundation}

\address{Department of Mathematics, Texas A\&M University, College Station, TX 77843-3368, USA}

\email{sunic@math.tamu.edu}

\begin{abstract}
We introduce the notion of a normal gallery, a gallery in which any
configuration of guards that visually covers the walls necessarily covers the entire
gallery. We show that any star gallery is normal and any gallery with at most
two reflex corners is normal. A polynomial time algorithm is provided deciding
if, for a given gallery and a finite set of positions within the gallery, there
exists a configuration of guards in some of these positions that visually covers
the walls, but not the entire gallery.       
\end{abstract}


\maketitle

\section{Introduction and main results}

An art gallery is a simple polygon (the boundary is a simple closed curve
consisting of a finite number of line segments) and a guard is a designated
point in the gallery. A guard $G$ in an art gallery $\Gamma$ visually covers
every point $A$ in the gallery for which the segment $GA$ is entirely within
$\Gamma$ (including the possibility that $GA$ intersects the boundary). A
configuration $F$ of guards that visually covers the gallery is a set of points
in the gallery such that every point in the gallery is visually covered by at
least one of the guards in $F$. It is known that a gallery with $n$ corners
can always be visually covered by a configuration of $\lfloor n/3 \rfloor$
guards~\cite{chvatal:n/3}. The subject of visual coverage of art galleries has
developed quite substantially since that paper, for instance by restricting or
extending the types of polygons considered (orthogonal
polygons~\cite{kahn-k-k:orthogonal,orourke:orthogonal,
sack-touissaint:orthogonal} or polygons with
holes~\cite{orourke:b-art,bjorling-s-s:with-holes,hoffmann-k-k:with-holes}),
considering guards that can move~\cite{orourke:mobile}, restricting the
positions that guards can occupy~\cite{bjorling:edge},  considering
combinations of existing variations, exploring higher dimensions, and so on. A
survey on the various achievements and directions of study can be found
in~\cite{urrutia:handbook-art}. 

We consider the relation between the visual coverage of the walls and the rest of the gallery.  

\begin{example}
It is known that there exists an art gallery $\Gamma$ and a configuration of
guards in $\Gamma$ that visually covers the walls of $\Gamma$, but does not
visually cover the entire gallery. One such gallery is exhibited in~\cite[page
4]{orourke:b-art} and we reproduce (a version of) it, denoted by $\Gamma_6$,
in Figure~\ref{f:6gear}. The guards $A$, $B$, and $C$ visually cover the
walls, but none of them covers $D$ (the dashed lines in the figure represent
the sight lines of $A$, $B$ and $C$)  
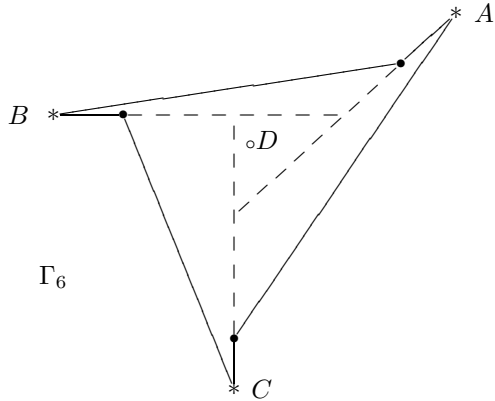
\begin{figure}[!ht]
\[
\xymatrix@R=13pt@C=15pt{
&&& &&&&&& &*{\ast} \ar@{-}[ld] \ar@{}[r]|{\textstyle{A}}&& 
\\
&&& &&&&&& *{\scriptstyle{\bullet}} \ar@{--}[dddlll]&&& 
\\
&&&*{\ast} \ar@{-}[rrrrrru] \ar@{-}[r] \ar@{}[l]|{\textstyle{B}}  & *{\scriptstyle{\bullet}} \ar@{--}[rrrr]&&&&&  &&& \\
&&& &&&& \ar@{}[ul]|{\scriptstyle{\circ}\textstyle{D}} && &&&\\
&&& &&&&&& &&&\\
&&& \Gamma_6 
\\
&&& &&&*{\scriptstyle{\bullet}} \ar@{-}[rrrruuuuuu] \ar@{--}[uuuu]&&& &&& 
\\
&&& &&&*{\ast} \ar@{-}[lluuuuu] \ar@{-}[u] \ar@{}[r]|{\textstyle{C}}&&& &&& 
}
\]
\caption{The guards $A$, $B$ and $C$ visually cover the walls, but not the
entire gallery}
\label{f:6gear}
\end{figure}
\end{example}

We are interested in conditions under which any configuration of guards that
visually covers the walls of a gallery necessarily covers the entire gallery.  

\begin{definition}
An art gallery $\Gamma$ is called \emph{normal} if any configuration of guards
in $\Gamma$ that visually covers the walls of the gallery necessarily covers the entire
gallery.  
\end{definition}

Note that, from a practical point of view, checking if a configuration of guards
covers the walls of a gallery may be an easier task than checking if it covers
the entire gallery. For instance, there might be situations in which it is
easy to implement sensor control (or some other type of control) along the
walls, while the access to the interior is restricted in some way. If we know
that the gallery is normal and we have a way of monitoring that the walls are
visually covered at all times, then, as long as the walls are covered, we are
sure that the guards are covering the entire gallery. In such a situation, for
as long as the walls are covered, we may even allow  relatively free movement
of the guards within the gallery without the need to constantly guide them, communicate with them, or
even have an information about their location.  

The terminology for normal galleries is modeled, albeit superficially, on the
terminology used for normal field extensions. Namely, the defining property of
normal field extensions is the ``one in - all in'' property of the roots,
while the defining property of the normal galleries is ``wall in - all in''.

We provide two (independent) sufficient conditions for a gallery to be normal.
Recall that a reflex corner in a gallery is a corner at which the interior
angle is greater than 180 degrees.  

\begin{theorem}\label{t:2reflex}
Every gallery with no more than two reflex corners is normal. 
\end{theorem}

A star gallery is a gallery that can be visually covered by a single guard
(more formally, there exists a point $P$
within the gallery such that, for every point $X$ in the gallery, the segment
$PX$ is entirely within the gallery). 

\begin{theorem}\label{t:star}
Every star gallery is normal. 
\end{theorem}

The proof of Theorem~\ref{t:2reflex} is based on the concept of support lines,
which is also helpful in establishing the following result.  

\begin{proposition}\label{p:2guards}
If the walls of a gallery are visually covered by one or two guards, then these
guards cover the entire gallery (even if the gallery is not normal).  
\end{proposition}

Deciding if a given gallery is normal seems to be a nontrivial task. However,
the task simplifies if the positions where the guards could be placed are
restricted.  

\begin{definition}
An art gallery $\Gamma$ is called \emph{normal with respect to a set} $A
\subseteq \Gamma$ if any configuration of guards in $A$ that visually covers the
walls of the gallery necessarily covers the entire gallery.  
\end{definition}

In particular, a gallery $\Gamma$ is normal if it is normal with respect to
$\Gamma$. Another special case of interest is when the positions of the guards
are restricted to the corners (or any particular finite set of positions within
the gallery). In this case there is an algorithm that decides, given as input a
gallery, if the given gallery is normal with respect to its set of corners. 
More generally (and more precisely), the following holds.  

\begin{theorem}\label{t:algorithm}
There exists an algorithm that decides, given as input a gallery $\Gamma$
with $n$ corners and a set $A$ of $m$ points in $\Gamma$, if $\Gamma$ is normal
with respect to $A$. The algorithm runs in $O(m^2n(m+n)\log m)$ time. 
\end{theorem}

The provided algorithm is based on the decomposition into visibility regions
with respect to a finite set of points, which is the preprocessing step in the
work of Bose, Lubiw, and Munro~\cite{bose-atal:visibility-queries} on
visibility queries. In particular, the complexity estimate is based on the
tight bound on the number of certain special regions in the visibility
decomposition (called sinks) provided in~\cite{bose-atal:visibility-queries}.
 

\section{Normal galleries}

\subsection{Support lines and galleries with no more than 2 reflex corners}

Before we are ready for the proof of Theorem~\ref{t:2reflex} we need to develop
the concept of a support line.  

For the duration of the present subsection we keep the following setting and
notation. Let $\Gamma$ be a gallery that is not normal and let $F$ be a
configuration of guards that visually covers the walls of the gallery without
covering the entire gallery. Let $R$ be the region in the gallery that is
not visually covered by any guard (the hidden region for the configuration $F$).
Since the walls of the gallery are visually covered by the guards, the region
$R$ is (an open set) in the interior of the gallery and the entire boundary
$\partial R$ of $R$, which consists of one or more polygonal lines, is visually
covered.   

Call the lines that support the boundary segments of the closure $\overline{R}$ the
support lines of $R$. Since each boundary segment of $\overline{R}$ is visually
covered by the guards, while $R$ is not, there must be a guard on each support
line.  

Each connected component of $R$ is the interior of a polygon and must have at
least three boundary segments.  

\begin{example}
Note that the hidden region $R$ does not have to be connected. For instance, the
gallery on the left in Figure~\ref{f:nostar} is not normal (guards are placed at
1,2,3,4,5) and the hidden region has two components.   
\begin{figure}[!ht]
\[
\begin{array}{ccc}
\xymatrix@R=8pt@C=8pt{
 *{\scriptstyle{\bullet}} \ar@{-}[rrr] &&& 
 *{\scriptstyle{\bullet}} \ar@{-}[d] &&&&&&
\\
 *{\ast} \ar@{-}[u] \ar@{}[d]|<<<<{\textstyle{1}}&
 *{\scriptstyle{\bullet}} \ar@{-}[l] \ar@{..}[rr] \ar@{..}[dr] &&
 *{\scriptstyle{\bullet}} \ar@{-}[dl] &&&&&&
\\
 &&
 *{\scriptstyle{\bullet}} \ar@{-}[ddrr] &&&&&*{\scriptstyle{\circ}} \ar@{-}[ddrr] &&
\\ 
 &
 *{\scriptstyle{\bullet}} \ar@{-}[uu] &&&&&
 \ar@{..}[d] \ar@{..}[dr] &&&
\\
 *{\ast} \ar@{-}[ur] \ar@{}[d]|<<<<{\textstyle{2}}&&&&
 *{\ast} \ar@{-}[r] \ar@{}[u]|<<<<{\textstyle{3}}&
 *{\scriptstyle{\bullet}} \ar@{-}[uurr] &
 *{.} \ar@{..}[r] &
 *{\scriptstyle{\bullet}} \ar@{-}[dl] &&
 *{\scriptstyle{\bullet}} \ar@{-}[dl] 
\\
 &&&&&&
 *{\scriptstyle{\bullet}} \ar@{-}[d] &&
 *{\ast} \ar@{-}[ul] \ar@{}[dr]|{\textstyle{5}}&
\\
 &&&&&&
 *{\ast} \ar@{-}[uullllll] \ar@{}[r]|{\textstyle{4}}&&&
}
&&
\xymatrix@R=8pt@C=8pt{
 *{\scriptstyle{\bullet}} \ar@{-}[rrrrrr] \ar@{-}[d]&& &&&& 
 *{\scriptstyle{\bullet}} \ar@{-}[ddddd]  
\\
 *{\scriptstyle{\bullet}} \ar@{-}[rr]&& 
 *{\scriptstyle{\circ}} \ar@{}[u]|{\textstyle{1}}&&
 *{\scriptstyle{\bullet}} \ar@{-}[ll] \ar@{-}[d] &&   
\\
 *{\scriptstyle{\bullet}} \ar@{-}[rr] \ar@{-}[ddd] &&
 *{\scriptstyle{\bullet}} \ar@{-}[d] &&
 *{\scriptstyle{\circ}}\ar@{}[r]|{\textstyle{2}} \ar@{-}[dd] 
 &&  
\\
 && 
 *{\scriptstyle{\circ}}\ar@{}[l]|{\textstyle{4}} \ar@{-}[d]& &&& 
\\
 &&
 *{\scriptstyle{\bullet}} \ar@{-}[r] & 
 *{\scriptstyle{\circ}} \ar@{-}[r] \ar@{}[d]|{\textstyle{3}} &
 *{\scriptstyle{\bullet}}&& 
\\
*{\scriptstyle{\bullet}} \ar@{-}[rrrrrr]  &&& &&& *{\scriptstyle{\bullet}}
}
\end{array}
\]
\caption{A gallery with a disconnected hidden region (left) and a normal gallery
that is not a star gallery (right)}
\label{f:nostar}
\end{figure}
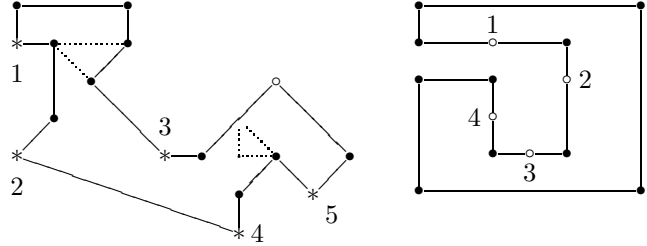

The same example shows that the boundary of the hidden region $R$ may in parts
coincide with the boundary of the gallery, that some support lines may have more
than one guard (2 and 3 are on the same support line, as well as 3 and 4), and
that some guards may be on more than one support line (guard 3) 
\end{example}

\begin{lemma}
The closure of every connected component of the hidden region is a convex polygon. 
\end{lemma}

\begin{proof}
Assume otherwise. Let $C$ be a reflex corner of the closure at the intersection
of the boundary segments $AC$ and $BC$ of the connected component $R_0$ of the
hidden region, and let $G_1$ be a guard on the support line of the side $AC$
visually covering the side $AC$ (see Figure~\ref{f:convex}).  
\begin{figure}[!ht]
\[
\xymatrix@R=10pt{
 && 
 *{\ast} \ar@{--}[d] \ar@{}[r]|{\textstyle{G_2}} && 
\\
 & 
 *{\scriptstyle{\circ}} \ar@{..}[r] & 
 *{\scriptstyle{\circ}} \ar@{..}[d] \ar@{}[r]|{\textstyle{B}} &&
\\
 *{\scriptstyle{\circ}} \ar@{..}[ur] && 
 *{\scriptstyle{\circ}} \ar@{..}[dr] \ar@{}[r]|{\textstyle{C}} &&
\\
 *{\scriptstyle{\circ}} \ar@{..}[u] &
 \ar@{}[ur]|<<<{\textstyle{R_0}} &&
 *{\scriptstyle{\circ}} \ar@{..}[dl] \ar@{--}[dr] \ar@{}[ur]|{\textstyle{A}} &
\\
 & 
 *{\scriptstyle{\circ}} \ar@{..}[ul] &
 *{\scriptstyle{\circ}} \ar@{..}[l] &&
 *{\ast} \ar@{}[u]|{\textstyle{G_1}}
}
\]
\caption{Closures of connected components of the hidden region are convex polygons}
\label{f:convex}
\end{figure}
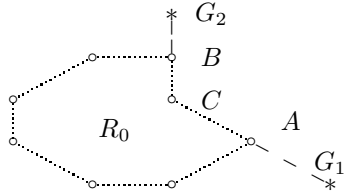
The sight line from $G_1$ towards $C$ cannot be interrupted at $C$, since the
immediate neighborhood of $C$ on the line $G_1C$ and inside the region $R_0$ is
inside the polygon. This would imply that this sight line extends inside $R_0$,
a contradiction.  
\end{proof}

Even though different support lines may share the same guard, this may not
happen if the support lines come from the same connected component of the hidden
region.  

\begin{lemma}
No guard can be on two different support lines bounding the same connected
component of the hidden region.  
\end{lemma}

\begin{proof}
Assume otherwise. Let $R_0$ be a connected component of the hidden region $R$
and let    $G$ be a guard on the lines supporting the boundary segments $AB$ and
$CD$ of $R_0$ (see Figure~\ref{f:no2}).  
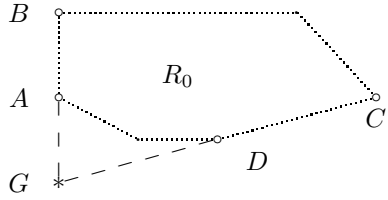
\begin{figure}[!ht]
\[
\xymatrix@R=10pt{
 &
 *{\scriptstyle{\circ}} \ar@{..}[rrr] \ar@{}[l]|{\textstyle{B}} &&&
 *{.} \ar@{..}[ddr]&  
 \\ 
 &&&&& 
 \\
 &
 *{\scriptstyle{\circ}} \ar@{..}[dr] \ar@{..}[uu] \ar@{}[l]|{\textstyle{A}} &&
 \ar@{}[ul]|{\textstyle{R_0}} && 
 *{\scriptstyle{\circ}} \ar@{}[d]|{\textstyle{C}} 
 \\
 && 
 *{.} \ar@{..}[r]& 
 *{\scriptstyle{\circ}} \ar@{..}[urr] \ar@{}[dr]|{\textstyle{D}} &&  
 \\
 &
 *{\ast} \ar@{--}[uu] \ar@{--}[urr] \ar@{}[l]|{\textstyle{G}}&&&& 
}
\]
\caption{No guard is on two support lines}
\label{f:no2}
\end{figure}
Since $R_0$ is convex the segment $AD$ is entirely within $R_0$, except for the
endpoints. No part of the walls of $\Gamma$ can be in the interior of the
triangle $GAD$ (since each of the three sides is within the polygon), which
shows that $G$ has an unobstructed view of some part of the region $R_0$, a
contradiction.  
\end{proof}

The following is an immediate corollary (note that Proposition~\ref{p:2guards}
is just a restatement of this corollary).
 
\begin{corollary}\label{c:ws3}
Let $\Gamma$ be a gallery that is not normal. Any configuration of guards in
$\Gamma$ that covers the walls, but not the entire gallery, must have at least
three members.  
\end{corollary}

\begin{proof}
Let $\Gamma$ be a gallery that is not normal and assume that a configuration of
guards that visually covers the walls, but not the entire gallery is given.  

Since the closure $\overline{R}_0$ of any connected component $R_0$ of the hidden
region $R$ has at least 3 sides, each support line contains a guard, and no
guard is on two support lines corresponding to the same connected component,
there are at least three guards in the given configuration.  
\end{proof}

\begin{proof}[Proof of Theorem~\ref{t:2reflex}]
Let $\Gamma$ be a gallery that is not normal and $F$ a configuration of guards
in $\Gamma$ that visually covers the walls but not the entire gallery. Consider
a connected component $R_0$ of the hidden region $R$.  

For each support line $\ell$ supporting a boundary segment $s$ of $R_0$ let
$G_\ell$ be a guard on $\ell$ that visually covers $s$ (if there is more
than one such guard select any of them). Recall that there are at least three
different support lines and that the guards chosen on different support lines
must be different.  

The reason the guard $G_\ell$ cannot see the hidden region $R_0$ must be a
presence of a reflex corner $C$ between $G_\ell$ and the boundary segment $s$ of
$R_0$ that is visually covered by $G_\ell$ such that the region $R_0$, and the
two neighboring corners $A$ and $B$ of $C$ are all on the same side of $\ell$
(including the possibility that one of $A$ and $B$ is on $\ell$) as indicated in
Figure~\ref{f:reflex}.  
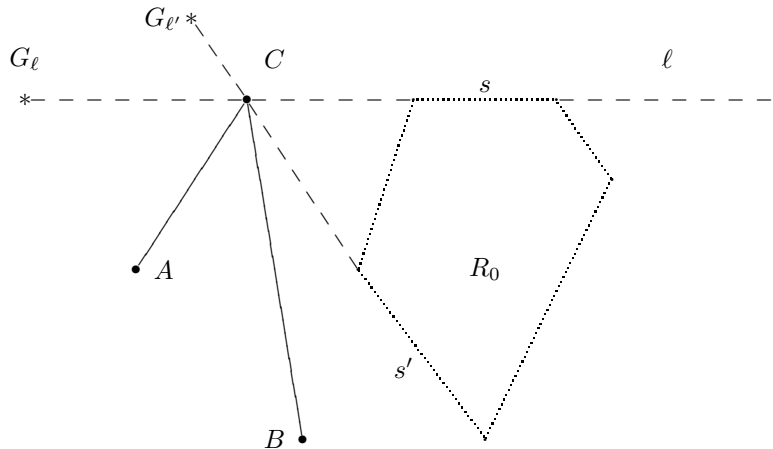
\begin{figure}[!ht]
\[
\xymatrix@C=15pt{
 &&&
 *{\ast} \ar@{--}[dr] \ar@{}[l]|{\textstyle{G_{\ell'}}} && &&&&& &&&
\\
 *{\ast} \ar@{--}[rrrr] \ar@{}[u]|{\textstyle{G_\ell}} &&&&
 *{\scriptstyle{\bullet}} \ar@{-}[ddll]\ar@{-}[ddddr]\ar@{--}[ddrr]\ar@{--}[rrr]
  \ar@{}[ur]|{\textstyle{C}}&&&
 *{.} \ar@{..}[rr]^{\textstyle{s}} &&
 *{.} \ar@{..}[dr] \ar@{--}[rrrr] &&
 \ar@{}[u]|{\textstyle{\ell}} &&
\\
 &&&&& &&&&&
 *{.} \ar@{..}[dddll] &&&
\\
 &&
 *{\scriptstyle{\bullet}} \ar@{}[r]|{\textstyle{A}} &&& &
 *{.} \ar@{..}[uur] &&R_0&& &&&
\\
 &&&&& &&&&& &&&
\\
 &&&&&
 *{\scriptstyle{\bullet}} \ar@{}[l]|{\textstyle{B}}  &&&
 *{.} \ar@{..}[uull]^{\textstyle{s'}}&& &&&
}
\]
\caption{Reflex corner obstructs the guard's view}
\label{f:reflex}
\end{figure}

Since $R_0$ is convex, there can be only one more support line of $R_0$ passing
through $C$, call it $\ell'$, but the view of the guard $G_{\ell'}$ on $R_0$
cannot be obstructed by the corner at $C$ since $R_0$ is not on the same side of $\ell'$ as $A$ and $B$. 

Therefore, a different reflex corner obstructs guards on different support
lines, and since there are at least three such lines (and guards), there are at
least three reflex corners in $\Gamma$.  
\end{proof}


\subsection{Star galleries and covers by convex polygons} 

\begin{proof}[Proof of Theorem~\ref{t:star}]
Let $\Gamma$ be a star gallery and $S$ be a point from which the entire gallery
is visible. Assume that a configuration of guards that visually covers the walls
of $\Gamma$ is given.  

We will show that for every point $W$ on the walls, there exists at least one
guard that visually covers the entire segment $SW$, which will show that the
guards visually cover the entire gallery.  

Given a point $W$ on the walls, let $G$ be a guard that visually covers it.
Since both $S$ and $W$ are visible from $G$ and the entire segment $SW$ is
within $P$, the entire segment $SW$ is visible from $G$.  
\end{proof}

It is tempting to conjecture that the converse of Theorem~\ref{t:star} is true,
but the next example shows that it is not.  

\begin{example}
Consider the gallery on the right in Figure~\ref{f:nostar}. It is not a star
gallery, but it is normal. Indeed, any configuration of guards that visually
covers the points $1$, $2$, $3$, and $4$ covers the entire gallery.  
\end{example}

The last example can be placed in a larger context. For a point $A$ in a
gallery, call the set of points visible from $A$ the \emph{view} of $A$.  

\begin{proposition}\label{p:convex-cover}
If $\Gamma$ is a gallery in which there exist a set $S$ of
points on the walls such that the points in $S$ have convex views and the
union of these views covers the entire gallery, then $\Gamma$ is normal. 
 
\end{proposition} 

\begin{proof}
This is clear since any guard that visually covers a point on the wall with
convex view belongs to this view and visually covers it.  
\end{proof}

\begin{remark}
We point out that there is a polynomial time algorithm that, given a gallery $\Gamma$ with $n$ corners, decides if there exist a set $S$ of points on the walls such that the points in $S$ have convex views and the 
union of these views covers the entire gallery. The main point to observe is that, for a given non-convex gallery,  it is sufficient to check the convexity of the views of the points in the set $S'$ consisting of the non-reflex corners on edges adjacent to reflex-corners and the midpoints of edges that are adjacent to two reflex corners, and then consider the union of the views of the points in $S'$ that have convex views. Details will be provided in a subsequent work. 
\end{remark}

By using Proposition~\ref{p:convex-cover}, one may easily show that there is no upper bound on
the number of guards needed to visually cover normal galleries nor on the number
of reflex corners for such galleries. Indeed, there are spiral,
right-angled galleries (in the spirit of the one shown on the right in
Figure~\ref{f:nostar}, but with more and more ``turns'') that require an 
arbitrarily large number of guards and have an arbitrarily large number of
reflex corners, while still being normal.  

The minimal number of guards that visually cover $\Gamma_6$ is 2. Since all star
galleries are normal, $\Gamma_6$ is, with respect to the number of guards that
can visually cover it, minimal among the galleries that are not normal. The next
example shows that $\Gamma_6$ is not minimal in a different sense.  

\begin{example}
The gallery $\Gamma_8$ in Figure~\ref{f:g8} is not normal, since guards at
corners $4$, $5$ and $8$ visually cover the walls of $\Gamma_8$, but do
not cover the entire gallery. 
\begin{figure}[!ht]
\[
\xymatrix@R=30pt@C=35pt{
 &
 *{\scriptstyle{\bullet}} \ar@{-}[rrrr] \ar@{}[l]|<<<<{\textstyle{1}} \ar@{--}[ddddrrrr] &
 \ar@{}[dl]|<<<<<{H_{4,5,6}}&\ar@{}[dr]|{H_{4,5}}&&
 *{\scriptstyle{\bullet}} \ar@{-}[dddd] \ar@{}[r]|<<<<{\textstyle{2}}&
\\
 &
 *{\ast} \ar@{-}[u] \ar@{}[l]|<<<<{\textstyle{8}} \ar@{}[urr]|<<<<<{H_{3,4,5,6}} & 
 *{\scriptstyle{\bullet}} \ar@{-}[l] \ar@{}[dl]|<<<<{\textstyle{7}} \ar@{--}[u] \ar@{--}[rrr]&
 \ar@{}[d]|{H_{4,5,8}}&&
 \ar@{}[ul]|<<<<{H_4} \ar@{}[r]|<<<<{W_{8,7}} & 
\\
 \ar@{}[d]|{\textstyle{\Gamma_8}} &&
 \ar@{}[r]|{H_{1,5,8}} &&
 \ar@{}[dr]|<<<{H_{4,8}} &&
\\
 &&
 *{\scriptstyle{\bullet}} \ar@{-}[uu] \ar@{}[l]|<<<<{\textstyle{6}} \ar@{--}[uuurrr] \ar@{--}[dd]&&&&
\\
 &
 \ar@{}[dr]|{H_{1,7,8}} &&
 \ar@{}[r]|{H_{1,8}} &&
 *{\scriptstyle{\bullet}} \ar@{-}[dr] \ar@{}[r]|<<<<{\textstyle{3}} \ar@{--}[d]&
\\
 *{\ast} \ar@{-}[uurr] \ar@{}[u]|<<<<{\textstyle{5}} &&&&&
 \ar@{}[ur]|<<<<{H_{1,2,8}}&
 *{\ast} \ar@{-}[llllll] \ar@{}[u]|<<<<{\textstyle{4}}
}
\]
\caption{A gallery that is not normal and can be decomposed into 3 convex polygons}
\label{f:g8}
\end{figure}
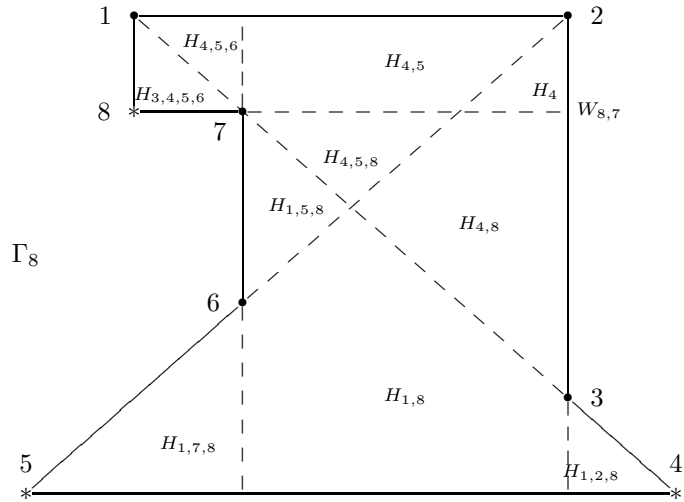
Indeed, none of them covers the points in the (open)
triangular region $H_{4,5,8}$ (the region is bounded by the dashed lines).
Note that $\Gamma_8$ can be decomposed into three convex polygons. On the
other hand, $\Gamma_6$ can be decomposed into 4, but not fewer 
than 4, convex polygons (none of the points $A,B,C,D$ are visible from each
other, showing that no two of them belong to a convex subset of $\Gamma_6$).
Since it is clear that galleries that can be decomposed into two convex polygons
are normal (in fact, they are star galleries; any point on the boundary between
the two convex polygons in the decomposition visually covers the entire
gallery), we see that $\Gamma_8$ is, with respect to the number of convex
polygons needed for its decomposition, minimal among the galleries that are not
normal.  
\end{example}

Note also that $\Gamma_6$ and $\Gamma_8$ are minimal among galleries that are
not normal both with respect to the number of reflex corners (three) and the
number of guards required to visually cover the walls but not the entire gallery
(three).  


\section{Galleries normal with respect to a set}

We begin by providing an example that shows that a gallery may be normal with
respect to a set without being normal.

\begin{example}\label{e:corners-not-enough}
The gallery $\Gamma_9$ in Figure~\ref{f:g9} is normal with respect to its corners, but it is not normal.  

Observe that any corner guard that can see corner 9 (any of 1, 2, 8 or 9) can see the entire
trapezoid 1289, any corner guard that can see the points near corner 8 on the
wall between 8 and 7 ( any of 2, 3, 7 or 8) can see the entire trapezoid 2378,
and any corner guard that can see corner 4 (any of 3, 4, 5, 6, or 7) can see the
entire trapezoid 4567. Since the trapezoids 1289, 2378, and 4567 cover the
entire gallery, this gallery is normal with respect to its corners.  

On the other hand, guards at $G$, 6, and 9 visually cover the walls without
covering the entire gallery.  
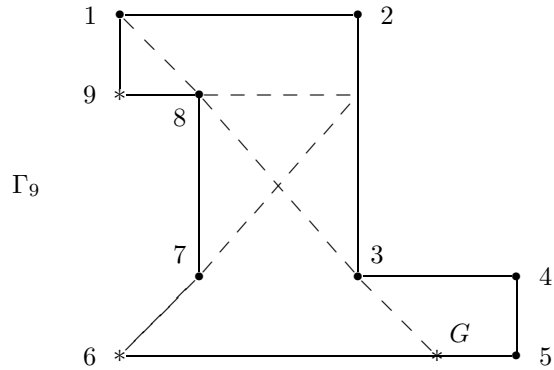
\begin{figure}
\[
\xymatrix{
 &
 *{\scriptstyle{\bullet}} \ar@{-}[rrr] \ar@{}[l]|<<<<{\textstyle{1}} &&&
 *{\scriptstyle{\bullet}} \ar@{-}[ddd] \ar@{}[r]|<<<<{\textstyle{2}}&&&
 \\
 &
 *{\ast} \ar@{-}[u] \ar@{}[l]|<<<<{\textstyle{9}} &
 *{\scriptstyle{\bullet}} \ar@{-}[l] \ar@{--}[rr] \ar@{--}[ul] \ar@{}[dl]|<<<<{\textstyle{8}} &&&&&
 \\
 \Gamma_9 &&&&&&&
 \\
 &&
 *{\scriptstyle{\bullet}} \ar@{-}[uu] \ar@{--}[uurr] \ar@{}[ul]|<<<<{\textstyle{7}} &&
 *{\scriptstyle{\bullet}} \ar@{-}[rr] \ar@{--}[uull] \ar@{}[ur]|<<<<{\textstyle{3}} &&
 *{\scriptstyle{\bullet}} \ar@{-}[d] \ar@{}[r]|<<<<{\textstyle{4}} &
 \\
 &
 *{\ast} \ar@{-}[ur] \ar@{}[l]|<<<<{\textstyle{6}}
 &&&&
 *{\ast} \ar@{--}[ul] \ar@{}[ur]|<<<<{\textstyle{G}}&
 *{\scriptstyle{\bullet}} \ar@{-}[lllll] \ar@{}[r]|<<<<{\textstyle{5}} &
}
\]
\caption{A gallery that is normal with respect to its corners, but is not normal}
\label{f:g9}
\end{figure}

\end{example}

Before we provide the polynomial time algorithm announced in
Theorem~\ref{t:algorithm} let us point out a rather simple approach that leads
to an exponential time solution. Namely, for a given point $P$ in a finite set
$A$ of $m$ points in a gallery with $n$ corners, there exists an algorithm of
time complexity $O(n)$ that determines the view of $P$ (the first such algorithm
for finding the visibility polygon of a point was given by ElGindy and
Avis~\cite{elgindy-a:view}; see also the work of Lee~\cite{lee:view}, and Joe
and Simpson~\cite{joe-s:view}). The same algorithm can be easily modified to
determine, in the same time, the portion of the walls that can be seen from
$P$. Thus, in time $O(mn)$ we may determine both the view and the wall view of
every point in $A$. At this moment, a straightforward approach would be to
check (potentially) all $2^m-1$ nonempty subsets of $A$ in search of a subset
that visually covers the walls without covering the entire gallery, but this
is certainly not an approach we want to follow. The search time for the right
subset of $A$ can be reduced considerably by utilizing the fact that the
gallery is simple and planar, which reduces the number of possible candidates
to polynomially many, and can be reduced even further by considering only a
certain class of candidates (corresponding to minimal visibility regions). 


\subsection{Decomposition into visibility regions}

Let $\Gamma$ be a gallery with $n$ corners and $A$ a set of $m$ points within the
gallery. The results and notions we are using are based on the work of Bose, Lubiw, and Munro~\cite{bose-atal:visibility-queries}. 

A pair of points $(P,C)$ is called \emph{feasible} if $P \neq C$, $P$ is a
point in $A$, $C$ is a reflex corner that is visible from $P$, and the two
walls meeting at $C$ are in the same half-plane with respect to 
$PC$ (by definition, each of the two half-planes includes the line $PC$). For
each feasibe pair $(P,C)$, draw the segment $CW_{P,C}$ on the line $PC$, where
$C$ is between $P$ and $W_{P,C}$ and $W_{P,C}$ is the furthest point on the
walls of $\Gamma$ visible from $P$.
Following~\cite{bose-atal:visibility-queries}, the segment $CW_{P,C}$ is
called the window of $P$ with base $C$.  

\begin{example}
Consider the gallery $\Gamma_8$ in Figure~\ref{f:g8} and let $A$
be the set of corners. The point $W_{8,7}$ is clearly indicated in the figure.
Further, $W_{3,7}=W_{4,7}=1$, $W_{5,6}=2$, and so on. 

As another example, consider the gallery $\Gamma_9$ in Figure~\ref{f:g9}. In
this gallery $W_{8,3}=G$, the pair (4,8) is not feasible because 8 is not
visible from 4, and the pair (6,3) is not feasible because 2 and 4 are on
different sides of the line 6-3.  
\end{example} 

The collection of all windows together with the gallery walls, decomposes the
gallery $\Gamma$ into a finite number of regions called \emph{visibility
regions} with respect to $A$. The decomposition is called the \emph{visibility
decomposition} of $\Gamma$ with respect to $A$. Since there can be no more
than $mn$ feasible pairs the number of regions into which $\Gamma$ is
decomposed is $O((mn)^2)$. However, using the planarity of the structure and a
more careful analysis, Bose, Lubiw, and Munro give a better estimate (they
also provide examples showing that their estimate is sharp).  

\begin{theorem}[Bose, Lubiw, Munro~\cite{bose-atal:visibility-queries}, Theorem 7]\label{t:number-regions}
The number of regions in the visibility decomposition of a gallery $\Gamma$
with $n$ corners with respect to a set $A$ of $m$ points is $O(m^2n)$.  
\end{theorem} 

The significance of the regions in the visibility decomposition for our purposes is derived from the following observation. 

\begin{lemma}[Bose, Lubiw, Munro~\cite{bose-atal:visibility-queries}, Lemma~19]\label{l:vr-defined}
Any two points in the interior of any visibility region in the visibility
decomposition of $\Gamma$ with respect to $A$ can be seen by the same subset
of points in $A$.  
\end{lemma} 

\begin{example}
The visibility regions with respect to the corners  of $\Gamma_8$ are indicated
in Figure~\ref{f:g8}. The boundaries of the regions are represented by dashed
lines. The notation $H_{4,5,8}$ is used as an indicator that this is the
region that is hidden from the corners 4, 5 and 8 (and visible from any other
corner).    
\end{example}


\subsection{Sinks}\label{ss:sink}

We define some special regions, called sinks, in the visibility decomposition.
These regions have minimal visibility with respect to $A$ and play a crucial
role in the work of Bose, Lubiw, and Munro. As
in~\cite{bose-atal:visibility-queries}, we assume that no three distinct
points chosen among the corners of the gallery and the set $A$ are collinear
(in fact, it is sufficient to assume that no two windows are collinear).  

For any region $R$ in the visibility decomposition, let $V(R)$ be the subset of
points in $A$ that visually cover the interior of the region $R$ (note that, by
Lemma~\ref{l:vr-defined}, a point in $A$ either visually controls the entire interior of the region $R$ or it does not control any point in it). 

\begin{definition}\label{d:sinks}
Let $\Gamma$ be a gallery with $n$ corners and $A$ a set of $m$ points within
$\Gamma$. A region $R$ in the visibility decomposition of $\Gamma$ with
respect to $A$ is called a \emph{sink} if, for every region $R'$ that shares a
common boundary edge with $R$, the set $V(R')$ contains the set $V(R)$.  
\end{definition} 

The dual graph to the visibility decomposition is defined as follows. Every
region in the decomposition is represented by a vertex and a directed edge
from the vertex representing $R$ to the vertex representing $R'$ is placed
whenever $R$ and $R'$ share a common boundary edge and $V(R)$ contains
$V(R')$. The non-collinearity condition ensures that the interiors of the
regions $R$ and $R'$ can be seen by the same subset of points in $A$ except
for a single point (the set $V(R)$ has one point more than $V(R')$ and this is
the point $P$ defining the window $PW_{P,C}$ supporting the common boundary
edge of $R$ and $R'$). Thus, any two vertices representing adjacent regions in
the decomposition are also adjacent in the dual graph, the graph is acyclic,
the sinks are precisely the regions corresponding to graph theoretic sinks in
the graph (vertices without outgoing edges), and no two sink regions share a
boundary edge~\cite{bose-atal:visibility-queries}. The dual graph leads, in the work of Bose, Lubiw, and Munro, to
a structure that can be used to recover the sets $V(R)$ for any visibility
region $R$ without keeping the full information for every region in the memory
($V(R)$ is only memorized for the sinks). A particularly useful fact about
sinks leading to good time complexity estimates in their work, as well as in
ours, is their low count (compared to the number of all regions).  

\begin{theorem}[Bose, Lubiw, Munro~\cite{bose-atal:visibility-queries}, Theorem 8]\label{t:number-sinks}
The number of sinks in the visibility decomposition of a gallery $\Gamma$ with
$n$ corners with respect to a set $A$ of $m$ points is $O(m(m+n))$.  
\end{theorem}


\subsection{The algorithm} 

The idea of the algorithm is simply to construct the visibility
decomposition and then, for each sink region $R$ in the decomposition,
consider the set $\overline{V}(R)$, the complement of $V(R)$ in $A$,
consisting of the vertices in $A$ for which $R$ is hidden, and check if guards
placed at all points in $\overline{V}(R)$ visually cover the walls. 

We first prove that it is indeed enough to check the sets $\overline{V}(R)$ only for sink regions $R$. 

\begin{lemma}\label{l:sink-enough}
A gallery $\Gamma$ is not normal with respect to a finite set $A$ if and only
if there exists a sink region $R$ in the visibility decomposition of $\Gamma$
with respect to $A$ such that guards placed at all points in $\overline{V}(R)$
visually cover the walls of the gallery.  
\end{lemma}

\begin{proof}
If such a sink region $R$ exists, the guards placed at all points in
$\overline{V}(R)$ visually control the walls and none of them controls any
point inside $R$, which, by definition, implies that $\Gamma$ is not normal
with respect to $A$.  

Assume that $\Gamma$ is not normal with respect to $A$. Then there exists a
configuration of guards at some subset $F$ of $A$ that visually controls the
walls, but not the entire gallery. This means that there exists an open set
within $\Gamma$ that is hidden to all points in $F$. This open set must
contain an interior point of some visibility region $R$ in the visibility
decomposition of $\Gamma$ with respect to $A$. Since none of the points in $F$
visually controls $R$, $F$ is a subset of $\overline{V}(R)$. If we enlarge $F$
to $\overline{V}(R)$, guards placed in all points in $\overline{V}(R)$ still
have no visual control of the region $R$ and they still control the walls
(since the guards in $F$ already do). If $R$ is a sink we are done. Otherwise,
there is an adjacent region $R'$ such that $V(R') \subsetneq V(R)$ and,
consequently, $\overline{V}(R) \subsetneq \overline{V}(R')$. We may then
enlarge the configuration of guards to the set of all points in
$\overline{V}(R')$. This enlarged configuration of guards has no visual
control of the region $R'$, but still controls the walls (since the guards in
$\overline{V}(R)$ already do). If $R'$ is a sink we are done. Otherwise we
continue the procedure by moving to an adjacent region $R''$ with an even
larger set $\overline{V}(R'')$. The finiteness of the visibility decomposition
(the finiteness of $A$) and the fact that we are adding guards at every step
ensures that this procedure must eventually end in a sink region.  
\end{proof}

\begin{proof}[Proof of Theorem~\ref{t:algorithm}]
The input of the algorithm is $\Gamma$, a gallery with $n$ walls, and $A$, a set
of $m$ points within the gallery. The corners of $\Gamma$ and the points in $A$ are
given by their coordinates in the plane. The corners are given in a sequence
that indicates their order along the boundary of the gallery (say counterclockwise). The output is
NORMAL, if the gallery $\Gamma$ is normal with respect to $A$, and NOT NORMAL,
otherwise. The algorithm proceeds through three steps (phases) described below, along
with time estimates. The steps are sequential, so the time complexity of the
entire algorithm is equal to the maximal time complexity of the individual
steps. \\

\begin{description}

\item[Step 1] \emph{Construct the visibility decomposition, identify the sink regions, and calculate $V(R)$ for each sink region $R$} \\
This step is, essentially, the preprocessing step described in Section 4
in~\cite{bose-atal:visibility-queries}. It uses the visibility polygon
algorithm (ElGindy and Avis~\cite{elgindy-a:view}, Lee~\cite{lee:view}, Joe
and Simpson~\cite{joe-s:view}) and the Bentley-Ottmann Algorithm for segment
intersections~\cite{bentley-ottmann:intersections} (see~\cite[Chapter
2]{deberg-atal:book-comp-geom} for a more current treatment). As indicated
in~\cite[Theorem 9]{bose-atal:visibility-queries}, this step takes $O(m^2(m+n)
\log n)$ time. \\ 

\item[Step 2] \emph{Determine the portion of the wall covered by each point in $A$} \\
For every point $P$ in $A$ determine the portion $W(P)$ of the gallery walls
that is visually covered by $P$. The output $W(P)$ should be given as follows.
Fix a corner of the gallery, and identify the boundary of the gallery with the
interval $[0,L]$, where $L$ is the total length of the gallery walls, 0
represents the chosen corner, and the point $x$ in the interval $[0,L]$
represents the point at distance $x$ along the walls as they are traversed in
the counterclockwise direction. The portion of the walls $W(P)$ visible to $P$
is represented as the union of $O(n)$ subintervals of the interval $[0,L]$.
Moreover, the visibility polygon algorithm provides this union in a sorted
manner, i.e., $W(P)$ consists of $O(n)$ subintervals of $[0,L]$ that are given
as a sorted list of endpoints. Since there are $m$ points in $A$ and the
visibility polygon algorithm requires $O(n)$ time for each point in $A$, this
step can be completed in $O(mn)$ time. \\ 

\item[Step 3] \emph{For each sink region $R$, check if $\overline{V}(R)$ controls the walls} \\
Consider a fixed sink region $R$. The set $\overline{V}(R)$ can be determined
in $O(m)$ time (since $A$ has $m$ elements and $V(R)$ is already known from
Step 1). Consider the union $\cup_{P \in \overline{V}(R)} W(P)$. We need to
check if this union covers the entire interval $[0,L]$. Each of the $O(m)$
wall portions $W(P)$, for $P$ in $\overline{V}(R)$, comes as a sorted list of
$O(n)$ subintervals of $[0,L]$. All these lists can be merged in time $O(mn
\log m)$ into a single sorted list of $O(mn)$ endpoints (for each endpoint one
only needs to keep its value in the interval $[0,L]$ and the information if it
is a left or a right endpoint; all endpoints are sorted by their value on
$[0,L]$, and for endpoints that have the same value the left endpoints are
considered smaller than the right endpoints). At this point, we may use Klee's
algorithm~\cite{klee:intervals} (see also~\cite[Chapter 8]{preparata:book})
that determines the measure of the union of  intervals in linear time with
respect to the number of intervals, as long as the endpoints are already
presorted (this is precisely why we first perform the merge sort indicated
above). Thus, after the presorting is completed, it may be checked in $O(mn)$
time if $\cup_{P \in \overline{V}(R)} W(P)$ is equal to $[0.L]$, i.e., if
guards placed at all points in $\overline{V}(R)$ cover the walls of the
gallery. Since there are $O(m(m+n))$ sink regions, the merging of the
intervals takes $O(mn \log m)$ time and  Klee's algorithm takes $O(mn)$ time,
this whole step takes $O(m(m+n))O((mn \log m) + mn)  = O(m^2n(m+n) \log m)$
time.  

If, for some sink region $R$, the union $\cup_{P \in \overline{V}(R)} W(P)$
turns out to be the whole interval $[0,L]$ stop and report NOT NORMAL.
Otherwise, after all sink regions are checked, stop and report NORMAL.  

\end{description}

\end{proof} 

\begin{remark}
Since the time complexity of the algorithm is carried by Step 3, there is no
need to try to improve the bounds in Step 1 by replacing the Bentley-Ottmann
Algorithm by any of the faster algorithms for segment intersections and map
overlays such as Chazelle-Edelsbrunner~\cite{chazelle-e:intersections} or
Balaban~\cite{balaban:intersections}. The time complexity of Step 3 is only
affected by the number of sink regions and the complexity of the merge sort of
several presorted lists. Since the estimate $O(m(m+n))$ for the number of sink
regions is shown to be sharp in~\cite{bose-atal:visibility-queries}, it seems
that there is not much room for improvements (unless an entirely different approach is
taken). 
\end{remark}

\section{Concluding Remarks}

In this work, we defined the notion of a normal gallery, a gallery in which
every configuration of guards that visually controls the walls necessarily
controls the interior. We established several sufficient conditions for a
gallery to be normal and provided an algorithm, running in polynomial time,
that checks if a given gallery is normal with respect to a given finite set of
positions within the gallery.  

We mention some natural follow up problems/questions. 

Given that a full characterization of normal galleries seems to
be a difficult task, there are two ways to proceed. One is to replace
the attempt to characterize by an attempt to find good sufficient and/or
necessary conditions and the other is to limit the domain
on which a characterization is sought. We briefly discuss both approaches. 

\subsection{Sufficient conditions} 

As mentioned in the introduction, there are situations in which checking if a
configuration of guards visually covers the entire gallery may be impractical,
while checking if they cover the walls may be relatively easy to implement. If
we know ahead of time that the gallery is normal, we know that it is actually
sufficient to check and ensure only the visual coverage of the walls. This
motivates a search for a wide range of sufficient conditions for normality
that are simple to verify (three such simple sufficient conditions are already
given here in Theorem~\ref{t:2reflex}, Theorem~\ref{t:star} and
Proposition~\ref{p:convex-cover}). 

\subsection{Normality in classes of galleries}

Since all galleries with at most two reflex corners are normal, and some
galleries with three reflex corners are not normal, a natural place to start is
to try to characterize the class of normal galleries with exactly three reflex
corners. 

In case a sufficiently simple and useful characterization is not possible, one
may try to find an algorithm that decides, given a gallery with three reflex
corners as input, if the given gallery is normal. 

Similarly, since all galleries that can be decomposed into two convex
parts are normal, one may try to characterize the normal galleries that can be
decomposed into three convex parts. 

Clearly, such questions (characterization and/or algorithmic
decidability of normality) may be posed within any other well defined class of
galleries. 

\subsection{Size and location of witness sets} 

Assume a gallery $\Gamma$ with $n$ corners is not normal. Call the smallest size of a configuration of guards that visually covers the walls, but not the entire gallery, the witness set size of $\Gamma$, and denote it by $wss(\Gamma)$.  Corollary~\ref{c:ws3} states that, for every non-normal gallery $\Gamma$, $wss(\Gamma) \geq 3$. It would be interesting to establish upper bounds on the maximum witness set size 
\[
 wss(n)=\max \{wss(\Gamma) \mid \Gamma\text{ a non-normal gallery with } n \text{ corners}\}. 
\]
in non-normal galleries in terms of the number of corners $n$. 

It is desirable to limit the search for witness sets to particular regions or even particular locations, if and when possible. For instance, is it true that every non-normal gallery has at least one witness set within the gallery walls? Note that Example~\ref{e:corners-not-enough} shows that the search for a witness set cannot be limited to the corners of the gallery. On the other hand, it would be interesting to effectively characterize the non-normal galleries that have witness sets within the set of corners.

\subsection*{Acknowledgments}

The author would like to thank Joseph O'Rourke for his comments and advice, and
the referees for their input, which led to a greatly improved version
of the text.


\end{document}